\documentclass[12pt]{amsart}
\usepackage{amsmath,amssymb,amsfonts,amsthm,amsopn}
\usepackage{pb-diagram}
\newcommand{\tfa}{time-frequency analysis}

\newcommand{\tf}{time-frequency}

\newcommand{\modsp}{modulation space}

\newtheorem{theorem}{Theorem}[section]
\newtheorem{lemma}[theorem]{Lemma}
\newtheorem{corollary}[theorem]{Corollary}
\newtheorem{proposition}[theorem]{Proposition}
\newtheorem{definition}[theorem]{Definition}

\newtheorem{remark}[theorem]{Remark}

\theoremstyle{definition}

\newcommand{\beqa}{\begin{eqnarray*}}
\newcommand{\eeqa}{\end{eqnarray*}}

\newcommand{\field}[1]{\mathbb{#1}}
\newcommand{\bR}{\field{R}}        %  real numbers
\newcommand{\bN}{\field{N}}        %  natural numbers
\newcommand{\bZ}{\field{Z}}        %  whole numbers
\def\G{\mathcal{G}}

\def\la{\lambda}

\def\th{\theta}

\def\cS{\mathcal{S}}

\def\cG{\mathcal{G}}
\def\cM{\mathcal{M}}

\def\cC{\mathcal{C}}

\def\a{\alpha}

\def\rd{\bR^n}

\def\rdd{{\bR^{2n}}}

\def\lrd{L^2(\rd)}

\def\zd{\bZ^n}

\def\intrd{\int_{\rd}}
\def\intrdd{\int_{\rdd}}

\def\R{\Big))}

\def\<{\left<}
\def\>{\Big)>}

\def\mv1{M_v^1}

\def\Mmpq{M_m^{p,q}}
\def\phas{(x,\xi)}
\def\mn{(m,n)}
\def\mn'{(m',n')}

\def\a{\alpha}
\def\b{\beta}

\def\t{\tau}

\def\R{\mathbb{R}}
\def\Ren{\mathbb{R}^n}
\def\Renn{\mathbb{R}^{2n}}

\def\sch{\mathcal{S}}

\def\Tau{\mathcal{T}}
\def\tauhz0{\widehat{\mathcal{T}}^\hbar(z_0)}
\def\tauhz{\widehat{\mathcal{T}}^\hbar(z)}

\def\Sn2{S_{2}(L^{2}(\Ren))}
\def\S1{S_{1}(L^{2}(\Ren))}
\def\sig00{\sigma_{0,0}}
\def\th{\widehat{\mathcal{T}}^\hbar}
\def\t{\widehat{\mathcal{T}}}
\def\st{\widehat{S_t}}
\def\sht{\widehat{S_t}^\hbar}

\def\la{\langle}
\def\ra{\rangle}

\newcommand{\Ha}{\widehat{H(t)}} % hamiltonian
\begin{document}
%[Semi-classical Time-frequency Analysis]	
\title {Semi-classical Time-frequency Analysis and Applications}
\author{Elena Cordero}
\address{Universit\`a di Torino, Dipartimento di
Matematica, via Carlo Alberto 10, 10123 Torino, Italy}
%\curraddr{\"{y} }
\email{elena.cordero@unito.it}
%\thanks{}
\author{Maurice de Gosson}
\address{University of Vienna, Faculty of Mathematics,
Oskar-Morgenstern-Platz 1 A-1090 Wien, Austria}
%\curraddr{\"{y} }
\email{maurice.de.gosson@univie.ac.at}
%\thanks{}
\author{Fabio Nicola}
\address{Dipartimento di Scienze Matematiche, Politecnico di Torino, corso
Duca degli Abruzzi 24, 10129 Torino, Italy}
%\curraddr{\"{y} }
\email{fabio.nicola@polito.it}
%\thanks{}  
\date{}

\begin{abstract}
This work represents a first systematic attempt to  create a common ground for   semi-classical and  time-frequency  analysis. These two different areas combined together provide interesting outcomes in terms of Schr\"odinger type equations. Indeed, continuity results of both Schr\"odinger propagators and their asymptotic solutions are obtained on $\hbar$-dependent Banach spaces, the semi-classical version of the well-known modulation spaces. Moreover, their operator norm is controlled by a constant independent of the Planck's constant $\hbar$. The main tool in our investigation is the joint application of standard approximation techniques from semi-classical analysis and a generalized version of Gabor frames, dependent of the parameter $\hbar$. Continuity properties of more general Fourier integral operators (FIOs) and their sparsity are also investigated. 
\end{abstract}

\maketitle

\section{Introduction}
This paper represents a first systematic attempt to set up a joint  framework for   semi-classical and time-frequency analysis. There are many excellent 
contributions on wave packet decompositions in semi-classical analysis (cf., e.g.,  \cite{ComberscureRobert2012, DiSj, hagedorn1980semiclassical, hagedorn1981semiclassical, Zworski} and references therein). The main motivation for this topic is quantum mechanics: the basic theme is to understand the relationships between dynamical systems and the behavior of solutions to Schr\"{o}dinger equations with a small positive parameter $\hbar\in (0,1]$, the Planck's constant, in other words,  how classical mechanics is a limit of quantum mechanics. 
The main tool for this understanding is the use of coherent states. According to Gilmore-Perelomov \cite{Gilmore}, a coherent state system is an orbit for an irreducible group action in an Hilbert space. In particular, the most well-known coherent states are obtained by the Weyl-Heisenberg
group action in $\lrd$ and the standard Gaussian 
\begin{equation}\label{psiGauss}
\psi_0(x)=2^{n/4}e^{-\pi |x|^2}.
\end{equation}
Recall that the $\hbar$-Weyl quantization of a function $H$ on the phase space $\rdd$ is formally defined by
\begin{align}\label{hWeyl}
\widehat{H}f(x)={\rm Op}^w_{\hbar}[H]f(x)&=(2\pi \hbar)^{-n}\intrdd e^{i\hbar^{-1}(x-y)\xi } H\Big(\frac{x+y}2,\xi\Big) f(y)\,dyd\xi\\
&=(2\pi)^{-n}\intrdd e^{i(x-y)\xi } H\Big(\frac{x+y}2,\hbar \xi\Big) f(y)\,dyd\xi
\end{align}
with $f$ in the Schwartz space $\cS(\rd)$. The function $H$ is called the $\hbar$-Weyl symbol of $\widehat{H}$. The Weyl-Heisenberg group action can be expressed by the Weyl quatization as follows.
For $z_0=(x_0,p_0)\in\rdd$, $f\in\cS(\rd)$, the Weyl-Heisenberg operator  $\widehat{\mathcal{T}}^\hbar(z_0)$ is given by
\begin{equation}\label{hT}
\widehat{\mathcal{T}}^\hbar(z_0)f(x)=Op^w_{\hbar}[e^{i\hbar^{-1}(p_0 x-x_0p)}]f(x)= e^{i\hbar^{-1}(p_0x-x_0p_0/2)}f(x-x_0).
\end{equation}
If $f=\psi_0$ in \eqref{psiGauss}, $\widehat{\mathcal{T}}^\hbar(z_0)\psi_0$, $z_0\in\rdd$, are the canonical coherent states.

It is then not surprising that semi-classical analysis
and time-frequency analysis are closely related because coherent states are the building blocks  for the so-called $\hbar$- Gabor frames, an extension of Gabor frames, the bricks of time-frequency analysis. \par

Gabor frames or Weyl-Heisenberg frames (the latter terminology was introduced in \cite{DGM86}), are applied in several different areas: for characterization of smoothness properties, in particular for the definition of modulation spaces (cf. Section $2$), for characterization of pseudodifferential operators and more generally Fourier integral operators (FIOs) (see Section $3$) and of course in signal processing (cf. \cite{30,31} and references therein).  They are widely employed in numerical analysis and used by engineers, but quite unknown among theoretical physicists. Although Gabor frames have been widely employed in the study of Schr\"odinger equations by two of us in the works \cite{fio6,fio5,CNR2015RevMP,fio2,CNR2015}, starting from the pioneering paper \cite{fio1}, there have not been displayed any direct connection with theoretical physics since the work of one of us \cite{deGossonACHA}. 
In that paper the definition of semi-classical Gabor frames,
 named $\hbar$-Gabor frames, first appears.  
  \begin{definition}\label{ee0}
  	Given a  lattice $\Lambda$ in $\rdd$ and a non-zero  function $g\in\lrd$, the system
  	$$\mathcal{G}^\hbar(g,\Lambda)=\{\tauhz g\,:\,z\in\Lambda\}
  	$$
  	is called a semi-classical Gabor frame or $\hbar$-Gabor frame if it is a frame for $\lrd$, that is there exist constants $0<A\leq  B$ such that
  	\begin{equation}\label{framedef}
  	A\|f\|_2^2\leq\sum_{z\in \Lambda}|\la f,\tauhz g\ra|^2\leq B\|f\|_2^2,\quad \forall f \in\lrd.
  	\end{equation}
  \end{definition}
  In particular, when $\hbar=(2\pi)^{-1}$,  we define $$\widehat{\mathcal{T}}(z)=\widehat{\mathcal{T}}^{(2\pi)^{-1}}(z)$$
   and we recapture the standard definition of a  Gabor frame simply by replacing the latter operator in \eqref{framedef} (see Section $2$ for details).\par

In this work we present the main features of semi-classical Gabor frames, continuing their investigation started in \cite{BBCNACHA2016,deGossonACHA} (see also \cite{nic} for an application to functional integration). \par 
The most well-known and used Banach spaces in \tfa \, are the so-called \modsp s and Wiener amalgam spaces \cite{F1}. Here we focus on the former spaces, whose norm can be interpreted as a measure of the joint \tf\, distribution of a signal $f$ in $\cS'(\rd)$, the space of tempered distributions. We refer to Section 2 for their exact definition and we recall that they are a scale of spaces comprehending 
the Sobolev spaces $H^s(\rd)$, the Shubin-Sobolev spaces $\mathcal{Q}_s$ \cite{BCG, Shubin91}, hence, in particular, the Hilbert space $L^2(\rd)$. 
For $\lambda>0$, we define the metaplectic operator $\widehat{D}_{\lambda}$  by
\begin{equation}\label{oplambda}\widehat{D}_{\lambda} f(x)=\lambda^{n/2} f(\lambda x),\,\quad f\in\lrd.
\end{equation}
If $M^{p,q}_m(\rd)$ denotes the standard modulation space defined in Subsection \ref{2.2}, and 
\begin{equation}\label{hdef}
h=2\pi \hbar
\end{equation}
we propose 
$$
M^{p,q,\hbar}_m(\rd)=\widehat{D}_{h^{-1/2}}M^{p,q}_m(\rd)
$$
as the semi-classical version of the  space $M^{p,q}_m(\rd)$.
The main motivation is that these spaces turn out to be the right Banach spaces for continuity results of both exact and  asymptotic solutions to semi-classical Schr\"odinger equations. Indeed we obtain norm estimates for such operators, {\it uniformly with respect to  the constant $\hbar$}. Namely, our focus is the Cauchy problem  in the semi-classical regime ($\hbar\to 0^+)$
\begin{equation}\label{C1}
\begin{cases} i \hbar
\partial_t u =\Ha u\\
u(s)=u_0,
\end{cases}
\end{equation}
where $t\in[0,T]$, $u_0\in\lrd$, $s\in[0,T]$ is the initial time, and the quantum Hamiltonian $\Ha$ is supposed to be the $\hbar$-Weyl quantization of the classical observable $H(t,z)$, with $z=\phas\in\rdd$. Such a symbol is supposed to satisfy the following hypothesis:\par
\medskip
{\bf Assumption (H)}. {\it The observable $H(t,z)$ is continuous with respect to $(t,z)\in [0,T]\times \rdd$ and smooth in $z$, satisfying
	\begin{equation}\label{ipotesi}
	|\partial^\alpha_z H(t,z)|\leq C_\alpha,\quad \forall |\alpha|\geq 2,\ z\in\rdd,\ t\in [0,T].
	\end{equation}
}
%This  is one of the main theme in semi-classical analysis. 

We decompose the initial datum $u_0$ in \eqref{C1} by means of a $\hbar$-Gabor frame whose atoms are Gaussian coherent states, construct asymptotic solutions for each of them,  so-called \emph{Gaussian beams}, and finally by superposition obtain the asymptotic solution (parametrix) to \eqref{C1}.\par 
Note that in our construction of the parametrix to \eqref{C1} we partially exploit the well-known tool of semi-classical approximation  by Taylor series \cite{ComberscureRobert2012}, already used in \cite{BBCNACHA2016}:  the approximate solution is searched as a finite sum of powers of $\hbar$, and the order of  approximation can be arbitrarily large. Let us underline that here time-frequency analysis comes in to play by using   $\hbar$-Gabor frames as coherent states.  By combining tools from both semi-classical and time-frequency analysis (the latter developed in Subsection \ref{2.1}), we attain the desired  continuity results (uniform with respect to $\hbar$) for the approximate solution on the semi-classical modulation spaces $
M^{p,q,\hbar}_m$. Moreover, we present precise estimates of the error term in such spaces, again uniform with respect to $\hbar$,  see Theorem \ref{mainteo} below, that can be regarded as the main result of this work.

All these issues let us claim that semi-classical Gabor frames and modulation spaces are, beyond $L^2$, the right Banach spaces to be used in  quantum mechanics. This assertion is confirmed by the study of a broader class of Schr\"odinger-type propagators, the  $\hbar$-dependent Fourier integral operators $\widehat{A}^\hbar:\cS(\rd)\to\cS'(\rd)$ in the class 
$FIO_\hbar(\chi,s)$, defined in terms of 
the decay properties of the entries of their so-called $\hbar$-Gabor matrix
\begin{equation}\label{eee}
|\langle \widehat{A}^\hbar\mathcal{T}^\hbar(z) g^h,\mathcal{T}^\hbar(w) g^h\rangle|\leq \frac{C_s}{\langle h^{-1/2}(w-\chi(z))\rangle^s},\quad z,w\in\rdd
\end{equation}
where $\la\cdot\ra=(1+|\cdot|^2)^{1/2}$, $\chi$ is a \emph{tame} symplectomorphism (introduced in \cite{fio1}, see also  \cite{fio3}) and $C_s>0$ is independent of $\hbar$; see Definition \ref{3.1} below.

In the following Section $3$ we investigate the main properties of such operators. In particular, we furnish again continuity properties on semi-classical modulation spaces and the so-called \emph{sparsity} properties of such operators. Roughly speaking, there are finitely many entries in \eqref{eee} that are not negligible, so that such operators can be efficiently represented  numerically.\par
This paper can be regarded as a first step of a project aiming at allowing semi-classical and time-frequency analysis to talk to each other. We believe that joining  the main features of these disciplines will provide an advancement in the understanding of both areas.

The work is organized as follows:  Section~2 contains the preliminary
notions from \tfa \, and the study of semi-classical Gabor frames and modulation spaces. In Section~3 we study the main properties of
semi-classical Fourier integral operators and provide an application to Schr\"odinger propagators. Section~4 contains the parametrix construction for  Schr\"odinger equations and exhibits the main result of this paper: Theorem \ref{mainteo}.

\section{Preliminaries and  \tfa \,tools}\label{Sect:Prel}
 The metaplectic group is denoted by $Mp(n)$. Consider $\widehat{S}\in Mp(n)$ with covering projection $\pi^\hbar: \widehat{S}\mapsto S \in {\rm Sp}(n,\R)$ the symplectic group of real $2n\times 2n$ matrices.  The appearance of the subscript $\hbar$ is due to the fact that to the $\hbar$-dependent operator $\widehat{V}_Pf(x)= e^{-i Px \cdot x/(2\hbar)}f(x)$ (chirp) corresponds the projection $\pi^\hbar (\widehat{V}_P)=V_P$, with $V_P=\begin{pmatrix} I_n&0_n\\-P&0_n\end{pmatrix}$, $P=P^T$,
   and to the Fourier transform $\widehat {J} f(x)=(2\pi i \hbar)^{-n/2}\intrd e^{-i x x'/\hbar} f(x') \,dx'$  corresponds    $\pi^\hbar(\widehat{J})=J$, defined by
   \begin{equation}
   \label{matriceJ}
   J=\begin{pmatrix} 0_n&I_n\\-I_n&0_n\end{pmatrix}.
   \end{equation}  For details see \cite[Appendix A]{deGossonACHA} and the books \cite{ComberscureRobert2012,deGossonbook}.
  In particular, for $\lambda>0$ we shall use the metaplectic operator $\widehat{D}_{\lambda}\in Mp(n)$  defined in
\eqref{oplambda}
  and whose projection is $\pi^\hbar (\widehat{D}_{\lambda})=D_{\lambda}$, the symplectic  matrix
 \begin{equation}
 \label{matricelambda}
 D_{\lambda}=\begin{pmatrix} \lambda^{-1} I_n& 0_n\\0_n&\lambda I_n\end{pmatrix}.
 \end{equation}
 In the sequel we shall often use the fundamental symplectic covariance formula
 \begin{equation}\label{CF}
 \tauhz \widehat{S}= \widehat{S}\th(S^{-1}z)\quad S\in {\rm Sp}(n,\R).
 \end{equation}

   \subsection{Semi-classical Gabor frames}\label{2.1}
 Consider a lattice $\Lambda$ in $\rdd$. For $g\in\lrd$, the Gabor system $$\mathcal{G}(g,\Lambda) =\{\widehat{\mathcal{T}}(z)g,z\in\Lambda\},$$
 (recall that $\t(z)=\t^{(2\pi)^{-1}}(z)$) is a Gabor frame for $\lrd$ if there exist
 constants $A,B> 0 $ such that for every $f\in\lrd$
 \begin{equation}
 \label{Frame_rel}
 A\|f\|^2_2 \leq \sum_{z\in\Lambda} |\la f,\widehat{\mathcal{T}}(z)g\ra|^2 \leq B\|f\|^2_2.
 \end{equation}
 Observe that,  up to a phase factor, $\widehat{\mathcal{T}}(z)$ is the so-called time-frequency (or phase-space) shift
 $$\widehat{\mathcal{T}}(z)f(t)=e^{-\pi i \xi x} e^{2\pi i \xi t} f(t-x)= e^{-\pi i \xi x}  M_{\xi}T_{x}f(t), \quad  z=(x,\xi),$$
 where translation and modulation operators are defined by
 $$T_{x}f(t) = f(t-x) \quad \mbox{and}\quad M_{\xi}f(t) = e^{2\pi i \xi  t} f(t).$$
 If $\eqref{Frame_rel}$ holds,
 then there exists a  $\gamma\in\lrd$ (so-called dual window),
 such that $\mathcal{G}(\gamma,\Lambda)$
 is a frame for $\lrd$ and every $f\in\lrd$ can be expanded as
 \begin{equation}\label{GabExp}
 f=\sum_{z\in\Lambda}\la f, \widehat{\mathcal{T}}(z)g \ra \widehat{\mathcal{T}}(z)\gamma= \sum_{z\in\Lambda}\la f,\widehat{\mathcal{T}}(z)\gamma\ra \widehat{\mathcal{T}}(z)g,
 \end{equation}
 with unconditional convergence in $\lrd$. 
 For $\Lambda=\a\zd\times\beta\zd$, with $\a,\beta>0$, it was proved by Lyubarski \cite{Lyu}, by Seip and Wallsten \cite{Seip,SW} in dimension $n=1$, and then easily extended using tensor product arguments the frame property for the Gaussian function $\psi_0$ in \eqref{psiGauss}.
 The result reads as follows. For $\Lambda=\a\zd\times\beta\zd$, we use the notation $\cG(g,\a,\beta)$ instead of $\cG(g,\Lambda)$ whenever it is more suitable.
   \begin{theorem}\label{Gaussians} The system $\cG(\psi_0,\a,\beta)$, where $\psi_0$  is the Gaussian function defined in \eqref{psiGauss}, is a frame for $L^2(\rd)$ if and only if $\a\b<1$.
   \end{theorem}
 Observe that, for any $g\in\lrd$, $r>0$,
 $$\widehat{D}_{r^{-1/2}}(T_{\a m} M_{\b n}g)=T_{\a r^{1/2} m} M_{\b r^{-1/2} n} \widehat{D}_{r^{-1/2}}g,
 $$
% $$\widehat{D}_{\hbar^{-1/2}}(T_{\a m} M_{\b n}g)=T_{\a \hbar^{1/2} m} M_{\b\hbar^{-1/2} n} \widehat{D}_{\hbar^{-1/2}}g,
% $$
 hence the Gabor system $\cG(\psi_0,\a,\b)$ is mapped by $\widehat{D}_{r^{-1/2}}$  to the Gabor system $\cG(\widehat{D}_{r^{-1/2}}\psi_0,\a r^{1/2},\b r^{-1/2})$.
 
 %=\cG( \phi^{\hbar}_0,\a\hbar^{1/2},\b\hbar^{-1/2}).$$
 Since $\widehat{D}_{r^{-1/2}}$ is a bijective isometry of $\lrd$,  $\cG(\psi_0,\a,\b)$ is a frame if and only if $\cG( \widehat{D}_{r^{-1/2}}\psi_0,\a r^{1/2},\b r^{-1/2})$ is, and the frame bounds coincide (in particular they are independent of $r$). 
 For $r=2\pi$, notice that  $\widehat{D}_{(2\pi)^{-1/2}}\psi_0=\phi_0$, where 
 \begin{equation}\label{Gaussian}
 \phi_0(x)=\pi^{-n/4}e^{-|x|^2/2},
 \end{equation}
 so that 
 $\cG(\phi_0,\a (2\pi)^{1/2},\b (2\pi)^{-1/2})$ is a Gabor frame if and only if $\a \b<1$ and with the same frame bounds as $\cG(\psi_0,\a,\b)$.\par
 We define its rescaled version by
 \begin{equation}\label{Gaussianh}
 \phi^{\hbar}_0(x)=\widehat{D}_{\hbar^{-1/2}}\phi_0(x)=(\pi \hbar)^{-n/4}e^{-|x|^2/(2\hbar)}.
 \end{equation}
Arguing as above we infer that  $\cG( \phi^{\hbar}_0,\a(2\pi \hbar)^{1/2},\b(2\pi\hbar)^{-1/2})$ is a Gabor frame if and only if $\a \b<1$ and the frame bounds coincide with those of $\cG(\psi_0,\a,\b)$. In particular, the frame bounds do not depend on the Plank constant $\hbar$.
 Indeed, defining $h$ as in \eqref{hdef},
 observe that 
 \begin{equation}\label{linkfp}
 \phi^{\hbar}_0=\widehat{D}_{h^{-1/2}}\psi_0
 \end{equation} hence the Gabor system $\cG( \phi^{\hbar}_0,\a h^{1/2},\b h^{-1/2})$ is the image of $\cG(\psi_0,\a,\b)$ under the dilation $\widehat{D}_{h^{-1/2}}$.\par
Let us also recall that  the Gabor frame $\cG(\psi_0,\a,\b)$ admits a dual window $\gamma_0$ (that is not the canonical one) such that $\gamma_0\in\cS(\rd)$, cf.  \cite{grochenig2009gabor}.\par
% This means that the Wexler-Raz biorthogonality relations (see, e.g. \cite[Theorem 7.3.1]{book}) are satisfied, that is 
% \begin{equation}\label{WR}
% (\a\b)^{-d} \left\langle \gamma_0, M_{\frac l\a} T_{\frac n\b} \psi_0\right\rangle=\delta_{l0}\delta_{n0},\quad l,n\in\zd.
% \end{equation}
% Observe that \eqref{WR} are equivalent to
% \begin{equation}\label{WR}
% ((\a h^{1/2})(\b h^{-1/2})^{-d} \left\langle \widehat{D}_{h^{1/2}}\gamma_0, M_{\frac l{\a h^{1/2}}} T_{\frac n{\b h^{-1/2}}} \phi^{\hbar}_0\right\rangle=\delta_{l0}\delta_{n0},\quad l,n\in\zd.
% \end{equation}
 From now onward we use the notation \begin{equation}\label{gh}
 g^h=\widehat{D}_{h^{-1/2}}g,
 \end{equation} for any function  $g\in \lrd$, so that the rescaled version of the dual window $\gamma_0$ reads
 \begin{equation}\label{fh}
 \gamma_0^h =\widehat{D}_{h^{-1/2}}\gamma_0.
 \end{equation} 
% We now introduce the definition of semi-classical Gabor frames, also called $\hbar$-Gabor frames by Maurice de Gosson, who first introduced thei definition in \cite{deGossonACHA}.
% \begin{definition}
% 	Given a  lattice $\Lambda$ in $\rdd$ and a non-zero  function $g\in\lrd$, the system
% 	$$\mathcal{G}^\hbar(g,\Lambda)=\{\tauhz g\,:\,z\in\Lambda\}
% 	$$
% 	is called a semi-classical Gabor frame or $\hbar$-Gabor frame if it is a frame for $\lrd$, that is there exist constants $0<A\leq  B$ such that
% 	\begin{equation}\label{framedef}
% 	A\|f\|_2^2\leq\sum_{z\in \Lambda}|\la f,\tauhz g\ra|^2\leq B\|f\|_2^2,\quad \forall f \in\lrd.
% 	\end{equation}
% \end{definition}
% In particular, when $\hbar=(2\pi)^{-1}$,  we obtain $\widehat{\mathcal{T}}(z)=\widehat{\mathcal{T}}^{(2\pi)^{-1}}(z)$ and we recapture the standard definition of a  Gabor frame. \par
The main properties of semi-classical Gabor frames, recalled in Definition \ref{ee0} were investigated in \cite[Proposition 2.3]{BBCNACHA2016} (see also \cite{deGossonACHA}). We recall Proposition 2.3 in \cite{BBCNACHA2016}, which shows how
  to  switch from a $\hbar$-Gabor frame to a standard Gabor frame and vice-versa:
  \begin{proposition}\label{P3}
 	The system $\mathcal{G}(\gamma,\Lambda)$ is a dual Gabor frame  of $\mathcal{G}(g,\Lambda)$ if and only if  $\mathcal{G}^\hbar(\gamma^h,h^{1/2}\,\Lambda)$  is a dual   $\hbar$-Gabor frame of the $\hbar$-Gabor frame $\mathcal{G}^\hbar(g^h,h^{1/2}\,\Lambda)$. Moreover, the frame bounds of $\mathcal{G}(g,\Lambda)$ and $\mathcal{G}^\hbar(g^h,h^{1/2}\,\Lambda)$ coincide, and the same holds for their dual frames.
 \end{proposition}
Using Proposition \ref{P3} for the Gaussian $\psi_0$ and its dual window $\gamma_0\in\cS(\rd)$, we can state:
\begin{corollary}\label{C3}
For $\Lambda=\a\zd\times\b\zd$, $\a\b<1$, the system  $\mathcal{G}^\hbar(\gamma_0^h,h^{1/2}\,\Lambda)$   is a dual  frame of the $\hbar$-Gabor frame $\mathcal{G}^\hbar(\phi_0^\hbar,h^{1/2}\,\Lambda)$. Moreover, the frame bounds of $\mathcal{G}^\hbar(\gamma_0^h,h^{1/2}\,\Lambda)$ are the same as those of $\mathcal{G}(\gamma_0,\Lambda)$. In particular, they are independent of $\hbar$.
\end{corollary}
\subsection{Semi-classical Modulation Spaces}\label{2.2}
The Banach spaces under our consideration will be a semi-classical version of modulation spaces.\par
Modulation  spaces were introduced by Feichtinger in \cite{F1} and have been widely employed over the last 20 years in the framework of time-frequency analysis.  For their definition, we need to recall the notion of weight functions on the \tf\ plane, which intervene in the  description of the  decay  properties of a function/distribution. We denote by  $v$  a
continuous, positive,  even, submultiplicative  weight function (in short, a
submultiplicative weight), i.e., $v(0)=1$, $v(z) = v(-z)$, and
$ v(z_1+z_2)\leq v(z_1)v(z_2)$, for all $z, z_1,z_2\in\Renn.$
A positive, even weight function $m$ on $\Renn$ is called  {\it
	v-moderate} if
$ m(z_1+z_2)\leq Cv(z_1)m(z_2)$  for all $z_1,z_2\in\Renn.$ We denote by $\mathcal{M}_v$ the class of all $v$-moderate weights. In this paper we will mainly work with the polynomial weights
\begin{equation}\label{vs}
v_s(z)=\la z\ra^s=(1+|z|^2)^{\frac s2},\quad z\in\rdd,\quad s\in\bR.
\end{equation} 
%Observe that $v_s$ is a $v_{|s|}$-moderate weight, for every $s\in\bR$. If  $m\in \mathcal{M}_{v_s}$, then there exist $C>0$ such that 
%\begin{equation}\label{weightcontrol}
%m(z)\leq Cv_sM(z),\quad \forall z\in\rdd.
%\end{equation} 
Let $f\in\mathcal{S}^{\prime }(\mathbb{R}^n)$. For $z=(x,\xi)\in\rdd$, we define the short-time
Fourier transform (STFT) of $f$ as 
\begin{equation}  \label{STFTdef}
V_gf(z)=\langle f,M_\xi T_xg\rangle=\int_{\mathbb{R}^n} f(t)\, {\overline {g(t-x)}} \, e^{-2\pi it \xi }\,dt.
\end{equation}
Observe that
\begin{equation}\label{STFTT}
|V_gf(z)|=|\langle f,\widehat{\Tau}(z)g\rangle|,\quad \forall z\in\rdd.
\end{equation}
Given a non-zero window $g\in\sch(\Ren)$, a $v$-moderate weight
function $m$ on $\Renn$, $1\leq p,q\leq
\infty$, the {\it
	modulation space} $M^{p,q}_m(\Ren)$ consists of all tempered
distributions $f\in\sch'(\Ren)$ such that the STFT $V_gf$ is in $L^{p,q}_m(\Renn )$
(weighted mixed-norm spaces), with norm 
$$
\|f\|_{M^{p,q}_m}=\|V_gf\|_{L^{p,q}_m}=\left(\int_{\Ren}
\left(\int_{\Ren}|V_gf(x,\xi)|^pm(x,\xi)^p\,
dx\right)^{q/p}d\xi\right)^{1/q}.  \,
$$
(Obvious modifications occur  when $p=\infty$ or $q=\infty$). If $p=q$, we write $M^p_m(\Ren)$ instead of $M^{p,p}_m(\Ren)$, and if $m(z)\equiv 1$ on $\Renn$, then we write $M^{p,q}(\Ren)$ and $M^p(\Ren)$ for $M^{p,q}_m(\Ren)$ and $M^{p,p}_m(\Ren)$.
Then  $\Mmpq (\Ren )$ is a Banach space whose definition is independent of the choice of the window $g$, in the sense that different  non-zero window functions yield equivalent  norms.

In the following we shall work with the a rescaled version of the Schwartz seminorms and modulation space norms, to make the corresponding spaces suitable for the analysis in the semi-classical regime, where we are looking for estimates independent of $\hbar$.\par
Namely, let us define, for $0<\hbar\leq1$,
\begin{equation}\label{esseh}
\cS^\hbar(\rd)=\widehat{D}_{h^{-1/2}}\cS(\rd)=\{f\in\cC^\infty(\rd)\,:\|\widehat{D}_{h^{1/2}}f\|_k<\infty\},
\end{equation}
where $\{\|\cdot\|_{k}\}_k$ $k\in\bN$, is the family of seminorms defining the Schwartz class $\cS(\rd)$:
$$\|f\|_{k}=\sum_{|\a|\leq k}\|(1+|\cdot|^k)\partial^\a f\|_\infty.$$
Clearly  $\cS^\hbar(\rd)=\cS(\rd)$ as sets.
\begin{definition}\label{mpqh}
For $m\in\mathcal{M}_v$, $1\leq p,q\leq \infty$, $0<\hbar\leq1$, we define
\begin{equation*}
M^{p,q,\hbar}_m(\rd)=\widehat{D}_{h^{-1/2}}M^{p,q}_m(\rd)=\{ f\in\cS'(\rd): \|f\|_{M^{p,q,\hbar}_m}:= \|\widehat{D}_{h^{1/2}}f\|_{M^{p,q}_m}<\infty \}.
\end{equation*}
\end{definition}
Using $(M^{p,q}_m)^*(\rd)=M^{p',q'}_{1/m}(\rd)$, for $1\leq p,q<\infty$,  we infer 
\begin{equation}\label{duality}
(M^{p,q,\hbar}_m)^*(\rd)=M^{p',q',1/\hbar}_{1/m}(\rd),
\end{equation}
for $1\leq p,q<\infty$.\par
 
Fix $g\in\cS(\rd)\setminus\{0\}$. For $f\in M^{p,q,\hbar}_m(\rd)$,  let us compute explicitly the  STFT of $\widehat{D}_{h^{1/2}}f$. Writing $z=(x,\xi)$, using $\widehat{\mathcal{T}}(x,\xi)=\widehat{\mathcal{T}}^\hbar(x,2\pi\hbar \xi)$ and the covariance property $\widehat{D}_{h^{-1/2}}\widehat{\mathcal{T}}^\hbar(z)=\widehat{\mathcal{T}}^\hbar ({D}_{h^{-1/2}}z)\widehat{D}_{h^{-1/2}}$, we obtain
\begin{align*}
|V_g (\widehat{D}_{h^{1/2}}f)(x,\xi)|&=|\la \widehat{D}_{h^{1/2}} f,\widehat{\mathcal{T}}(z)g\ra| =|\la \widehat{D}_{h^{1/2}} f, \widehat{\mathcal{T}}^\hbar(x,2\pi\hbar \xi)g\ra|\\
&=|\la f,\widehat{D}_{h^{-1/2}}\widehat{\mathcal{T}}^\hbar(x,2\pi\hbar \xi)g\ra|
=|\la f,\widehat{\mathcal{T}}^\hbar(h^{1/2}x,h^{1/2}\xi)\widehat{D}_{h^{-1/2}} g\ra|\\
&=|\la f,\widehat{\mathcal{T}}^\hbar(h^{1/2}x,h^{1/2}\xi) g^h\ra|
\end{align*}
(in the last row we used the notation \eqref{gh}). 
The previous computations motivate the definition of a semi-classical version of the STFT as follows.
\begin{definition}
Fix $g\in \cS(\rd)\setminus\{0\}$. We define the $\hbar$-short-time Fourier transform ($\hbar$-STFT)  $V_g^\hbar f$ of a function/distribution $f\in M^{p,q,\hbar}_m(\Ren)$ by
\begin{equation}\label{hSTFT}
V_g^\hbar f(z)= \la f,\widehat{\mathcal{T}}^\hbar(h^{1/2}z) g^h\ra, \quad \forall z\in\rdd.
\end{equation} 
\end{definition}
Fixed  $g\in\cS(\rd)\setminus\{0\}$, taking $f\in M^{p,q,\hbar}_m(\rd)$, we reckon
\begin{align*}
\|f\|_{M^{p,q,\hbar}_m}&=\|\widehat{D}_{h^{1/2}}f\|_{M^{p,q}_m}=\left(\intrd\left(\intrd |V_g(\widehat{D}_{h^{1/2}}f)\phas|^p m\phas^p\right)^{\frac q p}d \xi\right)^\frac 1q\\
&=\left(\intrd\left(\intrd |\la f,\widehat{\mathcal{T}}^\hbar(h^{1/2}x,h^{1/2}\xi) g^h\ra|^p m\phas^p\right)^{\frac q p}d \xi\right)^\frac 1q\\
&=\left(\intrd\left(\intrd |V_{g}^\hbar f|^p m\phas^p\right)^{\frac q p}d \xi\right)^\frac 1q.
\end{align*}
Hence, an equivalent definition of the semi-classical modulation spaces is as follows: Fix $g\in\cS(\rd)\setminus\{0\}$, then 

$$ M^{p,q,\hbar}_m(\rd)=\{ f\in\cS'(\rd)\,:\,  \|V_{g}^\hbar f\|_{L^{p,q}_m}<\infty\}.$$

As Gabor frames characterize modulation spaces, we shall show that $\hbar$-Gabor frames characterize semi-classical modulation spaces.\par Given a lattice $\Lambda\subset\rdd$, a Gabor frame $\mathcal{G}(g,\Lambda)$  with dual window $\gamma\in \lrd$, let us recall the coefficient (or analysis) operator $C_g: \lrd\to \ell^2(\Lambda)$, given by
$$C_g f=(\la f,\widehat{\Tau}(z)g\ra)_{z\in\Lambda}
$$
and the reconstruction (or synthesis) operator $D_\gamma: \ell^2(\Lambda)\to \lrd$:
$$R_\gamma c=\sum_{z\in\Lambda} c_{z}\widehat{\Tau}(z) \gamma,\quad c\in \ell^2(\Lambda).
$$
The related Gabor frame operator is defined by
$$S_{\gamma, g} f=\sum_{z\in\Lambda} \la f,\widehat{\Tau}(z)g\ra \widehat{\Tau}(z)\gamma.$$
The fact that $\gamma$ is a dual window of $g$ can be equivalently written as $S_{\gamma,g}=I$ on $\lrd$, that is the Gabor frame operator is the identity operator on $\lrd$.
The semi-classical analysis $C_g^\hbar$, synthesis $R_\gamma^\hbar$ and frame operators $S_{\gamma, g}^\hbar$ (also called $\hbar$-analysis, $\hbar$-synthesis and $\hbar$-frame operators) are
obtained simply by substituting $\widehat{\mathcal{T}}(z)$ with the Weyl-Heisenberg shift $\widehat{\mathcal{T}}^\hbar(z)$:
$$C_g^\hbar f=(\la f,\widehat{\mathcal{T}}^\hbar(z)g\ra)_{z\in\Lambda},
$$
$$R_\gamma^\hbar c=\sum_{z\in\Lambda} c_{z}\widehat{\Tau^\hbar}(z) \gamma,\quad c\in \ell^2(\Lambda).
$$
$$S_{\gamma, g}^\hbar f=\sum_{z\in\Lambda} \la f,\widehat{\Tau}^\hbar(z)g\ra \widehat{\Tau}^\hbar(z)\gamma.$$

  Provided the window $g,\gamma$ are smooth enough, we have the following characterization for modulation spaces (see, e.g. \cite[Corollary 12.2.6]{book}). 
%We use the notation $\tilde{m}(k,n)=m(\a k,\b n)$ for the restriction of the weight $m$ to the lattice $\Lambda=\a \zd\times\b \zd$.
\begin{proposition}\label{repf}
	Assume $g,\gamma\in M^1_v(\rd)$ and that $S_{\gamma,g}=I$ on $\lrd$. Then
	\begin{align*}
	f&=\sum_{k,n\in\zd}\la f,T_{\a k} M_{\b n}g\ra T_{\a k} M_{\b n}\gamma\\
	&=\sum_{k,n\in\zd}\la f,T_{\a k} M_{\b n}\gamma\ra T_{\a k} M_{\b n}g
	\end{align*}
with unconditional convergence in $M^{p,q}_m(\rd)$ if $1\leq p,q<\infty$ and weak$^*$ convergence in $M^\infty_{1/v}(\rd)$ otherwise. Furthermore, there are constants $0<A\leq B$  such that 
$$A\|f\|_{M^{p,q}_m(\rd)}\leq \| \la f, T_{\a k} M_{\b n}g\ra\|_{\ell^{p,q}_{{m}}(\Lambda)}\leq B \|f\|_{M^{p,q}_m(\rd)},\,\,\forall f \in M^{p,q}_m(\rd).
$$
(Similar estimates hold by replacing $g$ with $\gamma$).
\end{proposition}
For a given weight function $m\in \mathcal{M}_v$, let us introduce the notation
\begin{equation}\label{mhbar}
m_\hbar(z):= m(h^{-1/2}z).
\end{equation}
Then we obtain the following characterization for semi-classical modulation spaces. 
\begin{proposition}
	Under the assumptions of Proposition \ref{repf}, we have:\par
	$(i)$  The $\hbar$-analysis operator $C_g^\hbar$ is bounded from  $M^{p,q,\hbar}_m(\rd)$ to $\ell^{p,q}_{m_\hbar}(h^{1/2}\Lambda)$.\par
	$(ii)$ The $\hbar$-synthesis operator  $R_\gamma^\hbar$ is bounded from $\ell^{p,q}_{m_\hbar}(h^{1/2}\Lambda)$ to $M^{p,q,\hbar}_m(\rd)$.\par
	$(iii)$ The $\hbar$-frame operator satisfies
	$S_{\gamma, g}^\hbar=I$ on $M^{p,q,\hbar}_m(\rd)$. Namely, 
	\begin{align}\label{e1}
	f&=\sum_{z\in h^{1/2} \,\Lambda}\la f, \widehat{\mathcal{T}}^\hbar(z)g^h \ra \widehat{\mathcal{T}}^\hbar(z)\gamma^h\\
	&=\sum_{z\in h^{1/2} \,\Lambda}\la f, \widehat{\mathcal{T}}^\hbar(z)\gamma^h\ra \widehat{\mathcal{T}}^\hbar(z)g^h
	\end{align}
	with unconditional convergence in $M^{p,q,\hbar}_m(\rd)$ if $1\leq p,q<\infty$ and weak$^*$ convergence in $M^{\infty,1/\hbar}_{1/v}(\rd)$ otherwise. \par
	Furthermore, there are constants $0<A\leq B$, independent of $\hbar$, such that 
\begin{equation}\label{e2}
	A\|f\|_{M^{p,q,\hbar}_m(\rd)}\leq \| \la f, \widehat{\mathcal{T}}^\hbar(z)g^h\ra\|_{\ell^{p,q}_{{m_\hbar}}(h^{1/2}\Lambda)}\leq B \|f\|_{M^{p,q,\hbar}_m(\rd)},\,\,\forall f \in M^{p,q,\hbar}_m(\rd).
\end{equation}
	(Similar estimates hold by replacing $g^h$ with $\gamma^h$).
\end{proposition}
\begin{proof}
	 Observe that the metaplectic operators $\widehat{D}_{h^{1/2}}$ is an isometric isomorphism from  $M^{p,q,\hbar}_m(\rd)$  to  $M^{p,q}_m(\rd)$ with inverse $\widehat{D}_{h^{-1/2}}$.
	We treat the case $p,q<\infty$, the other case is similar. By Proposition \ref{repf}, for any $f\in M^{p,q,\hbar}_m(\rd)$, we can write, for $z=(x,\xi)$, and using the covariance property $\widehat{D}_{h^{-1/2}}\widehat{\mathcal{T}}^\hbar(z)=\widehat{\mathcal{T}}^\hbar ({D}_{h^{-1/2}}z)\widehat{D}_{h^{-1/2}}$,
	 \begin{align*}
	 \widehat{D}_{h^{1/2}}f&=\sum_{z\in \Lambda}\la  \widehat{D}_{h^{1/2}} f, \widehat{\mathcal{T}}(z)g \ra \widehat{\mathcal{T}}(z)\gamma \\
	 &=\sum_{z\in \Lambda}\la   f, \widehat{D}_{h^{-1/2}}\widehat{\mathcal{T}}(z)g \ra \widehat{\mathcal{T}} (z)\gamma \\
	 &=\sum_{(x,\xi)\in \Lambda}\la   f, \widehat{D}_{h^{-1/2}}\widehat{\mathcal{T}}^\hbar(x,2\pi\hbar \xi)g \ra \widehat{\mathcal{T}}^\hbar(x,2\pi\hbar \xi)\gamma \\
	 &= \sum_{(x,\xi)\in \Lambda}\la   f, \widehat{\mathcal{T}}^\hbar(h^{1/2}x,h^{1/2}\xi)\widehat{D}_{h^{-1/2}} g\ra \widehat{\mathcal{T}}^\hbar(x,h \xi)\gamma \\
	  &= \sum_{(x',\xi')\in  h^{1/2} \,\Lambda}\la   f, \widehat{\mathcal{T}}^\hbar(x',\xi') g^h\ra \widehat{\mathcal{T}}^\hbar( h^{-1/2}x', h^{1/2}\xi')\gamma 
	 \end{align*}
	 with unconditional convergence in $M^{p,q}_m(\rd)$. Now, applying the isomorphism $\widehat{D}_{h^{-1/2}}$ to both sides of the previous equalities and exploiting the covariance property of $\widehat{\mathcal{T}}^\hbar(z)$ again, we obtain \eqref{e1}, that is $S_{\gamma, g}^\hbar=I$ on $M^{p,q,\hbar}_m(\rd)$.\par
	 Since $g\in M^1_v$, it was proved in \cite[Theorem 12.2.3]{book} that the analysis operator $C_g$ is continuous from $M^{p,q}_m(\rd)$ into $\ell^{p,q}_{{m}}(\Lambda)$. Using similar arguments as before we infer that the following diagram is
	 commutative:
	 \[
	 \begin{diagram}
	 \node{M^{p,q,\hbar}_m(\rd)}\arrow{s,l}{C_g^\hbar} \arrow{c,t}{\widehat{D}_{h^{1/2}}}  \node{ M^{p,q}_m(\rd)}\arrow{s,r}{C_g}\\
	  \node{\ell^{p,q}_{m_\hbar}(h^{1/2}\Lambda)}\node{\ell^{p,q}_{{m}}(\Lambda)}\arrow{w,t}{\widehat{D}_{h^{-1/2}}}
	  \end{diagram}
	   \]
	 Notice that the lattice related to $C_g^\hbar$ is rescaled by $h^{1/2}$, whereas the domain of the weight $m$ is always  $\Lambda$, that is why we are introduced the notation $m_\hbar$. Indeed, for $z\in h^{1/2}\Lambda$, setting $z=h^{1/2}w$, we have $w\in\Lambda$ and the weight $m_\hbar(z)=m(h^{-1/2}h^{1/2}w)=m(w)$ does not depend on $h$ (or $\hbar$).
	 We underline that $$\|C_g^\hbar f\|_{\ell^{p,q}_{m_\hbar}(h^{1/2}\Lambda)}\leq C\|f\|_{M^{p,q,\hbar}_m(\rd)},\quad f\in M^{p,q,\hbar}_m(\rd)$$
	 for suitable constant $C>0$ independent of $\hbar$.\par
	 Arguing similarly we prove that also the diagram below is
	 	 commutative:
	  \[
	 	 \begin{diagram}
	 	 \node{\ell^{p,q}_{m_\hbar}(h^{1/2}\Lambda)}\arrow{s,l}{R_\gamma^\hbar} \arrow{c,t}{\widehat{D}_{h^{1/2}}}  \node{\ell^{p,q}_{{m}}(\Lambda)}\arrow{s,r}{R_\gamma}\\
	 	  \node{M^{p,q,\hbar}_m(\rd)}\node{M^{p,q}_m(\rd)}\arrow{w,t}{\widehat{D}_{h^{-1/2}}}
	 	  \end{diagram}
	 	   \]
	 and we have
	 $$\|R_\gamma^\hbar c\|_{M^{p,q,\hbar}_m(\rd)}\leq C\|c\|_{\ell^{p,q}_{m_\hbar}(h^{1/2}\Lambda)},\quad c\in 
	 \ell^{p,q}_{m_\hbar}(h^{1/2}\Lambda),$$
	 for a positive constant $C$ independent of $\hbar$.
	 From the continuity of the semi-classical analysis and synthesis operators we immediately obtain the norm equivalence \eqref{e2}. This concludes the proof.
	 \end{proof}
	 \section{Semi-classical FIOs}
	 In this section we study the sparsity and continuity on $\hbar$-modulation spaces of a class of Fourier integral operators (FIOs) arising as propagators for certain variable coefficients Schr\"odinger equations. Precisely, we introduce  the $\hbar$-version of  the Wiener algebras of FIOs studied in \cite{fio3} (see also \cite{fio1}).\par
	 Consider a {\it tame} symplectomorphism $\chi:\rdd\to\rdd$, that is\par\medskip 
	 \noindent {\it  (i)} $\chi$ is smooth, invertible,  and
	 preserves the symplectic form in $\rdd$, i.e., $dx\wedge d\xi= d
	 y\wedge d\eta$, if $(x,\xi)=\chi(y,\eta)$;
	 \par\medskip
	\noindent {\it  (ii)} $\chi$ satisfies 
	 \begin{equation}\label{agg10}
	| \partial^\alpha_z \chi(z)|\leq C_\alpha\quad |\alpha|\geq 1,\ z\in\rdd.
	 \end{equation}
	 \begin{definition}\label{3.1}
Consider a $\hbar$-Gabor frame $\G^\hbar(g,\Lambda)$  with
$g\in\cS(\rd)$, $s\geq 0$,  and $\chi$ be a {\it tame} symplectomorphism. We denote by $FIO_\hbar(\chi,s)$ the space of $\hbar$-dependent linear continuous operators $\widehat{A}^\hbar:\cS(\rd)\to\cS'(\rd)$ such that 
 \begin{equation}\label{agg11}
|\langle \widehat{A}^\hbar\mathcal{T}^\hbar(z) g^h,\mathcal{T}^\hbar(w) g^h\rangle|\leq \frac{C_s}{v_s(h^{-1/2}(w-\chi(z)))},\quad z,w\in\rdd
 \end{equation}
 for  some constant $C_s>0$ independent of $\hbar$.
 \end{definition}
 Observe that, for $\hbar=1/(2\pi)$ (hence $h=1$), we have $FIO_{1/(2\pi)}(\chi,s)=FIO(\chi,s)$, the class of FIOs studied in \cite{fio3}.
 %\subsection{Semi-classical sparsity in phase space}
 Using the formula 
 \[
 \widehat{D}_{h^{-1/2}}\mathcal{T}(z) g=\mathcal{T}^\hbar(h^{1/2} z)g^h
 \]
 we see that $\widehat{A}^\hbar$ satisfies \eqref{agg11} if and only if $\widehat{B}^\hbar:=\widehat{D}_{h^{1/2}}\widehat{A}^\hbar\widehat{D}_{h^{-1/2}}$ satisfies
 \begin{equation}\label{agg12}
|\langle \widehat{B}^\hbar \mathcal{T}(z) g,\mathcal{T}(w) g\rangle|\leq \frac{C_s}{v_s(w-\chi_\hbar(z))},\quad z,w\in\rdd
 \end{equation}
 for the same constant $C_s$, with
	 \[
	\chi_\hbar(z)=  \hbar^{-1/2}\chi(\hbar^{1/2}z).
	\]
We note that $\chi_\hbar$ satisfies, for every $\hbar\in (0,1]$, the same estimates \eqref{agg10} as $\chi$, {\it with the same constants $C_\alpha$}.
Following the notation in \cite{fio3}, this means that $\widehat{B}^\hbar\in FIO(\chi_\hbar,s)$. \par
Using the previous equivalence one can easily rephrase the main properties displayed by the class $FIO(\chi,s)$ in \cite{fio3} to its semi-classical generalization $FIO_\hbar(\chi,s)$. We list the main features below.
First, we present an equivalence between continuous decay conditions
and the decay of the discrete $\hbar$-Gabor
matrix for a linear operator $\widehat{A}^\hbar: \cS(\rd)\to\cS'(\rd)$.
\begin{theorem}\label{cara}
	Let $\widehat{A}^\hbar$ be a continuous linear operator $\cS(\rd)\to\cS'(\rd)$ and
	$\chi$ a tame symplectomorphism. Consider a $\hbar$-Gabor frame $\G^\hbar(g,\Lambda)$  with
	$g\in\cS(\rd)$ and $s\geq 0$. Then the following properties are
	equivalent. \par {\rm
		(i)} There exists $C_s>0$ independent of $\hbar$ such that \eqref{agg11} holds true.\\
	\indent{\rm (ii)}
	There exists $C_s>0$ independent of $\hbar$ such that
	\begin{equation}\label{unobis2}
	|\langle \widehat{A}^\hbar \mathcal{T}^\hbar(\lambda) g^h,\mathcal{T}^\hbar(\mu)g^h\rangle|\leq  \frac{C_s}{v_s(h^{-1/2}(\mu-\chi(\lambda)))},\quad \lambda,\mu\in \Lambda.
	\end{equation}
\end{theorem}
We call the infinite matrix $\{\langle \widehat{A}^\hbar \mathcal{T}^\hbar(\lambda) g^h,\mathcal{T}^\hbar(\mu)g^h\rangle\}_{\mu,\lambda\in \Lambda}$, the \emph{$\hbar$-Gabor matrix} of the operator $ \widehat{A}^\hbar$.\par 

Using similar arguments to \cite[Lemma 3.3]{fio3} we obtain: 
\begin{lemma}
	The definition of $FIO_\hbar(\chi ,s)$ is independent of the $\hbar$-Gabor frame $\cG^\hbar (g, \Lambda )$.
\end{lemma}
 \subsection{Continuity properties and sparsity} 
The FIOs in $FIO_\hbar(\chi ,s)$ display continuity results on semi-classical modulation spaces, with norms independent of $\hbar$, as explained in what follows.
\begin{theorem} \label{ee1} Consider $s>2n$, $m\in\cM_{v_r}$, with $0\leq r< s-2n$. Let 
	$\widehat{A}^\hbar \in FIO_\hbar(\chi ,s)$. Then $\widehat{A}^\hbar$ extends to a bounded operator from  $M^{p,\hbar}_{m\circ\chi}(\rd)$ to  
	$M^{p,\hbar}_{m}(\rd)$, for $1\leq p\leq \infty$, with operator norm uniformly bounded with respect to $\hbar$.
\end{theorem}
\begin{proof}
	Since the operator $\widehat{B}^\hbar$ in \eqref{agg12} satisfies $\widehat{B}^\hbar\in FIO(\chi,s )$, we can apply to $\widehat{B}^\hbar$ the continuity results proved in \cite[Theorem 4.1]{fio1}. Precisely, 
	an inspection of the proof of \cite[Theorem 4.1]{fio1} shows that such an operator $\widehat{B}^\hbar$ is bounded from $M^p_{m\circ \chi}(\Ren)$ into $M^p_{m}(\Ren)$ for every $1\leq p\leq\infty$. Moreover its operator norm is uniformly bounded with respect to $\hbar$, because the constant $C_s$ in \eqref{agg12} is independent of $\hbar$, and the estimates \eqref{agg10} for $\chi_\hbar$ hold, as already observed, uniformly with respect to $\hbar$. By the very definition of the space $M^{p,\hbar}_{m}(\Ren)$ we easily attain the result. 
\end{proof}

Due to the bilipschitz property of $\chi$,  $v_r\circ\chi\asymp v_r$,   hence the previous result for polynomial weights $v_r$ rewrites as follows.
 \begin{corollary}
 Consider  $\widehat{A}^\hbar \in FIO_\hbar(\chi ,s)$ and $0\leq r< s-2n$. Then $\widehat{A}^\hbar$ extends to a bounded operator on  $M^{p,\hbar}_{v_r}(\rd )$, for $1\leq p\leq \infty$, with operator norm uniformly bounded with respect to $\hbar$.
 \end{corollary}
 For   $p=2$, we have $M^{2}_{v_r}(\Ren)=\mathcal{Q}_r$, the Shubin-Sobolev spaces \cite{BCG, Shubin91}. Hence the previous issues hold also for the semi-classical version $\mathcal{Q}^\hbar_r:=M^{2,\hbar}_{v_r}(\Ren)$ of Shubin-Sobolev spaces.\par 
Given a tame symplectomorphism $\chi$, let us define
 $$FIO_\hbar(\chi)=\cap_{s\geq 0} FIO_\hbar(\chi,s).
 $$
If $\widehat{A}^\hbar\in  FIO_\hbar(\chi)$, then it satisfies \eqref{unobis2}
 for every $s\geq 0$. This means the operator $\widehat{B}^\hbar$ in \eqref{agg12} has a Gabor matrix highly concentrated along the graph of the map $\chi_\hbar$, or equivalently, the $\hbar$-Gabor matrix of $\widehat{A}^\hbar$ is concentrated along the graph of the map $\chi$. This matrix property is called \emph{sparsity} (cf. \cite{fio2}) and using Theorem \ref{ee1} we can state
\begin{proposition}\label{pro13}
Any $\hbar$-dependent operator $\widehat{A}^\hbar\in FIO_\hbar(\chi)$ is bounded on $M^{p,\hbar}_{v_r}(\Ren)$ for every $1\leq p\leq\infty$, $r\in\R$, with operator norm uniformly bounded with respect to $\hbar$.
\end{proposition}
\subsection{Applications to Schr\"odinger propagators} We can now consider an application to the sparsity and continuity in modulation spaces of Schr\"odinger propagators. \par
%Recall that the $\hbar$-Weyl quantization of a function $H$ on the phase space $\rdd$ is formally defined by
%\begin{align*}\label{hWeyl}
%\widehat{H}f(x)={\rm Op}^w_{\hbar}[H]u(x)&=(2\pi \hbar)^{-d}\intrdd e^{i\hbar^{-1}(x-y)\xi } H\Big(\frac{x+y}2,\xi\Big) f(y)\,dyd\xi\\
%&=(2\pi)^{-d}\intrdd e^{i(x-y)\xi } H\Big(\frac{x+y}2,\hbar \xi\Big) f(y)\,dyd\xi
%\end{align*}
%with $f$ in the Schwartz space $\cS(\rd)$. The function $H$ is called the $\hbar$-Weyl symbol of $\widehat{H}$.\par
We are interested in the Cauchy problem \eqref{C1}
with symbol $H(t,z)$, $z=\phas\in\rdd$, satisfying the {\bf Assumption (H)}. 
Now, denote by $\chi^{(t,s)}$, $t,s\in [0,T]$, the Hamiltonian flow with Hamiltonian function $H$, i.e. $(x_{t,s},\xi_{t,s})=\chi^{(t,s)}(x,\xi)$ satisfies

\[\dot{x}_{t,s} = \partial_\xi H(t,x_{t,s}, \xi_{t,s}),\quad \dot{\xi}_{t,s} = -\partial_x H(t,x_{t,s}, \xi_{t,s})
\]
with initial value $(x_{s,s},\xi_{s,s})=(x,\xi)$ (the dot denoting the derivative with respect to $t$). It is easily seen that the map $\chi^{(t,s)}$ satisfies estimates of the type \eqref{agg10} with constants $C_\alpha$ independent of $s,t\in [0,T]$.\par
Let finally $U(t,s)$ be the propagator for the equation in \eqref{C1}, so that $U(s,s)=I$. \par We have the following result.
\begin{proposition}\label{pro14}
The propagator $U(t,s)$ belongs to $FIO_\hbar(\chi^{(t,s)})$ uniformly with respect to $t,s\in [0,T]$. As a consequence $U(t,s)$ is bounded on $M^{p,\hbar}_{v_s}(\Ren)$ for every $1\leq p\leq\infty$, $s\in\R$, with operator norm uniformly bounded with respect to $\hbar\in (0,1]$, $s,t\in [0,T]$.
\end{proposition}
\begin{proof}
Consider the propagator 
$\tilde{U}(t,s)$ satisfying 
\begin{equation}\label{0metap}
 i \partial_t \tilde{U}(t,s) ={\rm Op}^w_1 [\hbar^{-1} H(t,\hbar^{1/2}x,\hbar^{1/2}\xi)] \tilde{U}(t,s),\quad\tilde{U}(s,s) =I,
 \end{equation}
 and compare it with $U(t,s)$, which satisfies
 \begin{equation}\label{0metap1}
 i \hbar \partial_t U(t,s) =\widehat{H(t,\cdot)} U(t,s)\quad U(s,s) =I.
 \end{equation}
The two are related by the formula
 \begin{equation}\label{0metap2}
U(t,s)=\widehat{D}_{\hbar^{-1/2}} \tilde{U}(t,s)\widehat{D}_{\hbar^{1/2}}.
\end{equation}
 Indeed we have
\begin{align*}
i\hbar\partial_t U(t,s)&=i\hbar \partial_t \widehat{D}_{\hbar^{-1/2}}\tilde{U}(t,s)\widehat{D}_{\hbar^{1/2}}\\
&= \widehat{D}_{\hbar^{-1/2}} {\rm Op}^w_1[H(t,\hbar^{1/2}x,\hbar^{1/2}\xi)] \tilde{U}(t,s) \widehat{D}_{\hbar^{1/2}}\\
&= {\rm Op}^w_1[H(t,x,\hbar _\xi)]\widehat{D}_{\hbar^{-1/2}}\tilde{U}(t,s)\widehat{D}_{\hbar^{1/2}}\\
&={\rm Op}^w_1[H(t,x,\hbar \xi)]U(t,s)\\
&=\widehat{H(t,\cdot)}U(t,s).
\end{align*}
Now, it was proved in \cite[Corollary 7.4]{tataru} that $\tilde{U}(t,s)$ satisfies an estimate of the form \eqref{agg12}, namely 
\[
|\langle \tilde{U}(t,s) \mathcal{T}(z) g,\mathcal{T}(w) g\rangle|\leq \frac{C_n}{v_n(w-\chi^{(t,s)}_\hbar(z))},\quad z,w\in\rdd
\]
for every $n\in\mathbb{N}$, where \[
\chi^{(t,s)}_\hbar(z)=\hbar^{-1/2}\chi^{(t,s)}(\hbar^{1/2}z)
\]
 is the flow corresponding to the Hamiltonian
\[
\hbar^{-1} H(t,\hbar^{1/2}z).
\]
As a consequence, $U(t,s)\in FIO_\hbar(\chi^{(t,s)})$ uniformly with respect to $s,t\in [0,T]$ and the desired continuity result follows from Proposition \ref{pro13}.
\end{proof}

\section{A parametrix construction for Schr\"{o}dinger equations}
One can construct approximate solutions to the problem \eqref{C1} by the the following construction. Let us first consider the case when the initial datum in a Gaussian function, possibly translated and modulated. From now on we fix $s=0$ in \eqref{C1} as initial time.\par
 Consider the solution $z_t=(x_t,\xi_t)$ to the Hamiltonian system
 \begin{equation}\label{HS} \dot{x}_t = \partial_p H(t,x_t, \xi_t),\quad \dot{\xi}_t = -\partial_x H(t,x_t, \xi_t)
 \end{equation}
 with initial value (at $t=0$) $z_0=(x_0,\xi_0)$. Let $\chi_t$ be the Hamiltonian flow defined by $H(t,z)$, hence \begin{equation}\label{flow}
  z_t=\chi_t(z_0).
 \end{equation} Define
 \begin{equation}\label{Hquad}
 H^{(l)}_{z_0}(t,z)=\sum_{|\gamma|=l}\frac1 {\gamma!}\partial^\gamma_z H(t,z_t)z^\gamma,\quad z\in\rdd,\ l\geq2.
 \end{equation}
 It is well-known that the solution to the corresponding operator Schr\"odinger equation with $\hbar=1$, namely
 \begin{equation}\label{metap}
 i \partial_t \widehat{S_t}(z_0) ={\rm Op}^w_1 [H^{(2)}_{z_0}(t)] \widehat{S_t}(z_0)\quad\widehat{S_0}(z_0) =I,
 \end{equation}
 is a metaplectic operator $\widehat{S_t}(z_0)$ corresponding to the symplectic matrix $S_t(z_0)$ via the metaplectic representation \cite{deGossonbook,ComberscureRobert2012}. In fact $S_t(z_0)$ is the (linear) Hamiltonian flow of $H^{(2)}_{z_0}(t)$ and therefore, as a consequence of \eqref{ipotesi}, the entries of the matrix $S_t(z_0)$ are bounded functions of $t\in[0,T]$ and $z_0\in\rdd$.\par
Now, an asymptotic solution to \eqref{C1}, modulo $O(\hbar^{(N+1)/2})$, $N\in\R$, with  initial value being the coherent state $\phi^{\hbar}_{z_0}$:
 \begin{equation}\label{C1z0}
\begin{cases} i \hbar
\partial_t u =\Ha u\\
u(0)=\phi^{\hbar}_{z_0},
\end{cases}
\end{equation}
is provided by the so-called Gaussian beam (see, e.g.,  \cite[Section 3]{BBCNACHA2016})
\begin{equation}\label{solapprox}
 \phi^{\hbar,N}_{z_0}(x)=e^{\frac {i} {\hbar}\delta(t,z_0)}\widehat{\mathcal{T}}^{\hbar}(z_t)\widehat{D}_{\hbar^{-1/2}} \widehat{S_t}(z_0)\sum_{j=0}^N \hbar^{j/2}b_j(t,x)\phi_0(x).
\end{equation}
Here $b_0(t,x)=1$ and for $j>0$, $b_j(t,x)$ is a polynomial in $x$, having coefficients depending on $t,z_0$ which are bounded (for details, see \cite[Section 3]{BBCNACHA2016}). The symmetrized action $\delta$ is defined by
\begin{equation}\label{symmaction}
\delta(t,z_0)=\int_0^t \Big(\frac12 \sigma(z_s,\dot{z}_s)-H(s,z_s)\Big)\,ds,
\end{equation}
with $\sigma$ being the standard symplectic form.\par
The functions $\phi^{\hbar,N}_{z_0}$ turn out to be approximate solutions in the sense that
\begin{multline}\label{agg2}
R^{(N)}_{z_0}(t,\cdot):=
(i\hbar \partial_t -\widehat{H(t)})\phi^{\hbar,N}_{z_0}(t,\cdot)\\
=-e^{\frac{i}{\hbar}\delta(t,z_0)}\th(z_t)\widehat{D}_{\hbar^{-1/2}}\st(z_0) \Big(
\sum_{l+k\geq N+3\atop {3\leq l\leq N+2 \atop 0\leq k\leq N}}\hbar^{(l+k)/2}{\rm Op}^w_1[H^{(l)}_{z_0}(t,S_t(z_0)z) ]b_k(t,\cdot)\phi_0\\
+\sum_{k=0}^N \hbar^{(N+3+k)/2}{\rm Op}^w_1[r^{(N+3)}_{z_0}(t,S_t(z_0)z)]b_k(t,\cdot)\phi_0\Big)
\end{multline}
is $O(\hbar^{(N+3)/2})$, where we set
\begin{equation}
\label{agg3}
r^{(N+3)}_{z_0}(t,z):=\frac{1}{(N+3)!}\sum_{|\gamma|=N+3} \int_0^1 \partial^\gamma_z H(t,z_t+\theta\hbar^{1/2} z) z^\gamma (1-\theta)^{N+2}\, d\theta.
\end{equation} 
In \cite{BBCNACHA2016} it was in fact introduced a parametrix via Gabor frames, valid for arbitrary $L^2$ initial data. Such a parametrix, having the previous Gaussian beams as building blocks, is constructed as follows.\par
For $\a\b<1$, $\Lambda=\alpha \bZ^n\times \beta \bZ^n$, consider the $\hbar$-Gabor frame $\G^\hbar(\phi_0^\hbar,h^{1/2}\Lambda)$. Let $\gamma_0^h$ be the dual window in $\cS(\rd)$ defined in \eqref{fh}.  For $N\geq0$, $t\in [0,T]$, the parametrix to \eqref{C1} is defined by
\begin{equation}\label{Gab_exp}
[U^{(N)}(t) f](t,\cdot)=\sum_{z_0\in h^{1/2}\Lambda}\langle f,\th(z_0)\gamma_0^h\rangle \phi_{z_0}^{\hbar,N}(t,\cdot).
\end{equation}
Observe that $U^{(N)}(0) f=f$.\par
This is our main result.
\begin{theorem}\label{mainteo}
Consider $s\in \bR$ and the weight function $v_s$ defined in \eqref{vs}. Under the {Assumption {\bf (H)}} and with the above notation, there exists a constant $C=C(T)$ independent of $\hbar$ such that, for every $f\in M^{p,\hbar}_{v_s}(\rd)$, $1\leq p\leq \infty$,
\begin{equation}\label{teoa0}
\|U^{(N)}(t)f\|_{M^{p,\hbar}_{v_s}(\rd)}\leq C\|f\|_{M^{p,\hbar}_{v_s}(\rd)}\quad \forall t\in [0,T]
\end{equation}
and
\begin{equation}\label{teob0}
\|(i\hbar \partial_t-\widehat{H(t)})U^{(N)}(t)f\|_{M^{p,\hbar}_{v_s}(\rd)}\leq C\hbar^{(N+3)/2}\|f\|_{M^{p,\hbar}_{v_s}(\rd)}\quad \forall t\in[0,T].
\end{equation}
If $U(t)$ denotes the exact propagator, for every $f\in M^{p,\hbar}_{v_s}(\rd)$,
\begin{equation}\label{teoc0}
\|(U^{(N)}(t)-U(t))f\|_{M^{p,\hbar}_{v_s}(\rd)}\leq C t \hbar^{(N+1)/2}\|f\|_{M^{p,\hbar}_{v_s}(\rd)} \quad \forall t\in[0,T].
\end{equation}
\end{theorem}
The proof relies on the following lemmas.
\begin{lemma}\label{3.2}
Consider $t\in [0,T]$, $z_0\in\rdd$, and the metaplectic operator $\st(z_0)$, solution to \eqref{metap}. Define \begin{equation}\label{sth}
 \sht(z_0)=\widehat{D}_{\hbar^{-1/2}}\st(z_0)\widehat{D}_{\hbar^{1/2}}.
\end{equation}
Then the metaplectic operator $\sht$ is continuous  from $\cS^\hbar$ to $\cS^\hbar$ with operator seminorms bounded with respect to $t\in[0,T]$, $z_0\in\rdd$, and independent of $\hbar$.
\end{lemma}
\begin{proof}
It was already shown in the proof of Theorem 3.4 in\cite{BBCNACHA2016} that 
 metaplectic operators $\st(z_0)$ are bounded $\cS(\rd)\to\cS(\rd)$, with the entries of the matrix $S_t(z_0)$ being bounded  functions of $t\in[0,T]$ and $z_0\in\rdd$. Hence the Schwartz seminorms are bounded with respect to $t,z_0$.
 Using $$\sht(z_0)=\widehat{D}_{\hbar^{-1/2}}\st(z_0)\widehat{D}_{\hbar^{1/2}}=
 \widehat{D}_{h^{-1/2}} \widehat{D}_{(2\pi)^{-1/2}}\st(z_0)\widehat{D}_{(2\pi)^{1/2}}\widehat{D}_{h^{1/2}},$$
 we immediately obtain the result.
\end{proof}	 
\begin{lemma}\label{3.3}
	Let $\mathcal{B}$ be a fixed bounded subset of $\cS^\hbar(\rd)$ (as a Fr\'echet space with the seminorms in \eqref{esseh}) and $g\in\mathcal{S}(\rd)$. Then for every $N\geq 0$ there exists a positive constant $C_N$ independent of $\hbar$ such that 
	\begin{equation}\label{dec}
|\la f, \th(z)g^h\ra|\leq \frac{C_N}{v_N(h^{-1/2}z)},  \quad\forall z\in\rdd
	\end{equation}
	for every $f\in\mathcal{B}$.
\end{lemma}
\begin{proof}
	Observe that the mapping $\widehat{D}_{h^{1/2}}: \cS^\hbar(\rd)\to \cS(\rd)$ is an  isomorphism. Hence, for $f\in\mathcal{B}\subset\cS^\hbar(\rd)$, 
	$\widehat{D}_{h^{1/2}}f$ belongs to a bounded subset of $\cS(\rd)$. As a consequence, for all $N\geq0$, there exists $C_N>0$ (independent of $\hbar$) such that 
	$$|\la \widehat{D}_{h^{1/2}}f, \widehat{\Tau}(z)g\ra|=|V_g (\widehat{D}_{h^{1/2}}f)(z)| \leq C_N v_{-N} (z),\quad z\in\rdd
	$$
because $V_g:\mathcal{S}(\rd)\to\mathcal{S}(\rdd)$ is continuous (see e.g., \cite[Theorem 11.2.5]{book}). \par
	Now $$\la \widehat{D}_{h^{1/2}}f, \widehat{\Tau}(z)g\ra=\la f ,\widehat{\Tau}^\hbar(h^{1/2}z)g^h\ra$$ and the claim follows.
\end{proof}

We now are ready to prove Theorem \ref{mainteo}.\par
\begin{proof}[Proof of Theorem \ref{mainteo}]
Let us first prove \eqref{teoa0}. We consider the case $p<\infty$, and we leave to the reader the case $p=\infty$, which is completely similar. For $f\in M^{p,\hbar}_{v_s}(\rd)$ we have, using \eqref{e2} with $g^h=\gamma_0^h$ defined in \eqref{fh},
	\begin{align}
\|U^{(N)}(t)f\|_{M^{p,\hbar}_{v_s}}^p&\asymp \sum_{w\in h^{1/2}\Lambda}|\la U^{(N)}(t)f, \widehat{\Tau}^\hbar(w)\gamma_0^h\ra|^p {v_s}(h^{-1/2}w)^p\nonumber\\
&= \sum_{w\in h^{1/2}\Lambda}|\la \sum_{z_0\in h^{1/2}\Lambda}\langle f,\th(z_0)\gamma_0^h\rangle \phi_{z_0}^{\hbar,N}(t,\cdot), \widehat{\Tau}^\hbar(w)\gamma_0^h\ra|^p {v_s}(h^{-1/2}w)^p\nonumber\\
&\leq \sum_{w\in h^{1/2}\Lambda}\Big(\sum_{z_0\in h^{1/2}\Lambda}|\la f ,\th(z_0)\gamma_0^h\rangle|\,| \la \phi_{z_0}^{\hbar,N}(t,\cdot), \widehat{\Tau}^\hbar(w)\gamma_0^h\ra|\Big)^p {v_s}(h^{-1/2}w)^p.\label{agg0}
	\end{align}
	Now, 
	\begin{align*}
	|\la  \phi_{z_0}^{\hbar,N}(t,\cdot), \widehat{\Tau}^\hbar(w)\gamma_0^h\ra|&\leq 
	\sum_{j=0}^N \hbar^{j/2}|\la \widehat{\mathcal{T}}^{\hbar}(z_t)\widehat{D}_{\hbar^{-1/2}} \widehat{S_t}(z_0)
	b_j(t,\cdot)\phi_0, \widehat{\Tau}^\hbar(w)\gamma_0^h\ra|\\
	&= \sum_{j=0}^N \hbar^{j/2}|\la \widehat{D}_{\hbar^{-1/2}} \widehat{S_t}(z_0)
	b_j(t,\cdot)\phi_0, \widehat{\Tau}^\hbar(w-z_t)\gamma_0^h\ra|\\
	&=\sum_{j=0}^N \hbar^{j/2}|\la  \widehat{S_t}^\hbar(z_0)
	b_j^h(t,(2\pi)^{1/2}\cdot)\psi_0^h, \widehat{\Tau}^\hbar(w-z_t)\gamma_0^h\ra|.\\
	\end{align*}
	Observe that $b_j(t,(2\pi)^{1/2}\cdot)\psi_0\in\cS(\rd)$ for every $j=0,\dots, N$ and so $b_j^h(t,(2\pi)^{1/2}\cdot)\psi_0^h\in\cS^{\hbar}(\rd)$
	% for every $j=0,\dots, N$,
	 with seminorms independent of $h$. Using Lemma \ref{3.2} we have  \[
	 \widehat{S_t}^\hbar(z_0)
	b_j^h(t,(2\pi)^{1/2}\cdot)\psi_0^h \in\cS^{\hbar}(\rd)
	\]
	 with seminorms uniformly bounded with respect to $h$. Hence,
	by Lemma \ref{3.3} we obtain
	$$|\la  \widehat{S_t}^\hbar(z_0)
	b_j^h(t,(2\pi)^{1/2}\cdot)\psi_0^h, \widehat{\Tau}^\hbar(w-z_t)\gamma_0^h\ra|\leq \frac{C_{\tilde{N}} }{v_{\tilde{N}}( h^{-1/2}(w-z_t))},\quad\forall z_0,w\in\rdd$$
for a constant $C_{\tilde{N}}>0$ independent of $\hbar$.\par
	These estimates yield, for every $\tilde{N}\geq0$,
	\begin{equation}\label{agg1}
		|\la  \phi_{z_0}^{\hbar,N}(t,\cdot), \widehat{\Tau}^\hbar(w)\gamma_0^h\ra|\leq \frac {C_{\tilde{N},N}}{ v_{\tilde{N}}(h^{-1/2}(w-z_t))}
		\end{equation}
	for a convenient $C_{\tilde{N},N}>0$ independent of $\hbar$.\par
	Now, by the second inequality in \eqref{e2} (again with $g^h=\gamma_0^h$) we have 
	\[
	\sum_{z_0\in h^{1/2}\Lambda}|\la f,\widehat{\Tau}^\hbar(z_0)\gamma_0^h\ra|^p{v_s}(h^{-1/2}z)^p\leq C \|f\|^p_{M^{p,\hbar}_{v_s}(\rd)},
	\] 
 for a suitable $C>0$. Using the $v_{|s|}$ moderateness of ${v_s}$, we can write ${v_s}(w)\leq C v_{|s|}(w-z_t){v_s}(z_t)$  for a convenient $C>0$. 
 Easily follows from the assumption \eqref{ipotesi} (see e.g. \cite{fio3}) that the map $z_0\to z_t$  is globally Lipschitz continuous and allows the equivalence 
 \begin{equation}\label{flowequiv}
 v_s(z_t)\asymp v_s(z_0)\quad \forall z_0\in\rdd,\,\forall t\in\bR.
 \end{equation}
 Hence \eqref{teoa0} will follow by combining \eqref{agg0}, \eqref{agg1} and \eqref{flowequiv}, provided that  the kernel 
	\begin{equation}\label{agg}
	K(w,z_0):=C_{\tilde{N},N}v_{-\tilde{N}+|s|}( h^{-1/2}(w-z_t))\asymp  C_{\tilde{N},N} v_{-\tilde{N}+|s|}(h^{-1/2}(\chi_t^{-1}(w)-z_0))
	\end{equation}
	(the equivalence is due to the bi-Lipschitz property of the Hamiltonian flow $z_0\to \chi_t(z_0)=z_t$  defined in \eqref{flow}) satisfies Schur's estimates
	\[
	\sup_{w\in h^{1/2}\Lambda}\sum_{z_0\in h^{1/2}\Lambda} |K(w,z_0)|<A_1,\qquad \sup_{z_0\in h^{1/2}\Lambda}\sum_{w\in h^{1/2}\Lambda}|K(w,z_0)|<A_2,
	\]
	with constants $A_1,A_2>0$ independent of $\hbar$. This is clearly the case if $\tilde{N}>2n+|s|$ by the very definition \eqref{agg}. For example, making the change of variables $z'=h^{-1/2}z_0-h^{-1/2}\chi^{-1}_{t}(w)$, we infer
	\begin{align*}\sum_{z_0\in h^{1/2}\Lambda} |K(w,z_0)|&\lesssim \sum_{z_0\in h^{1/2}\Lambda} v_{-\tilde{N}+|s|}(h^{-1/2}(\chi_t^{-1}(w)-z_0))\\
	&= \sum_{z'\in\Lambda-h^{-1/2}\chi^{-1}_{t}(w)} v_{-\tilde{N}+|s|}(z')< A_1
	\end{align*}
	for a constant $A_1>0$ independent of $w$, $h$ and $t$.\par

	The proof of \eqref{teob0} is similar: by the same argument  as in \eqref{agg0} we obtain
	\begin{align*}
\|(i\hbar \partial_t-\widehat{H(t)})&U^{(N)}(t)f\|_{M^{p,\hbar}_{v_s}(\rd)}^p\\
&\leq \sum_{w\in h^{1/2}\Lambda}\Big(\sum_{z_0\in h^{1/2}\Lambda}|\la f ,\th(z_0)\gamma_0^h\rangle|\,| \la R^{(N)}_{z_0}(t,\cdot), \widehat{\Tau}^\hbar(w)\gamma_0^h\ra| \Big)^pv_s(h^{-1/2}w)^p.
\end{align*}
where $ R^{(N)}_{z_0}(t,\cdot)$ is defined in \eqref{agg2}; observe that in \eqref{agg2} we can take out the factor $\hbar^{(N+3)/2}$ which appears in \eqref{teob0}. The pseudodifferential operators appearing in \eqref{agg2} are easily treated observing that, by \eqref{ipotesi}, their symbols together with their derivatives of every order are dominated by $C\langle z\rangle^{N+3}$, with a constant $C$ independent of $t\in[0,T]$ and $z\in\rdd$, so that  they are continuous on $\cS(\rd)$ with operator  seminorms uniformly bounded. The proof therefore goes on as that of \eqref{teoa0}.\par
	Finally let us prove the formula \eqref{teoc0}. We use Duhamel's formula: if $U(t,s)$ is the exact propagator, with $U(s,s)=I$, and $U(t)=U(t,0)$, we have
\[
U^{(N)}(t)f-U(t)f=-\frac{i}{\hbar}\int_0^t U(t,s)\Big(i\hbar\partial_s-\widehat{H(s)} \Big)U^{(N)}(s)f\,ds.
\]

Now the desired result follows at once from \eqref{teob0} and Minkowski inequality, using the fact that $U(t,s)$ is bounded $M^{p,\hbar}_{v_s}(\rd)\to M^{p,\hbar}_{v_s}(\rd)$ with norm uniformly bounded with respect to $\hbar, t,s$, when $0<\hbar\leq 1$, $0\leq s\leq t\leq T$. This was in fact proved in Proposition \ref{pro14}.  \end{proof}
\begin{remark}
(i) Observe that the results of the previous theorem hold true if we replace $v_s$ with a radial weight $m\in\mathcal{M}_{v_s}$. Indeed, we need that $m(z_t)\asymp m(z_0)$, for every $z_0\in\rdd$.\par
(ii) Some technique (interchanging series) of this proof does not apply for $M^{p,q,\hbar}_m(\Ren)$, with $p\not=q$. 
In fact, the result does not hold for  $M^{p,q,\hbar}_m(\Ren)$, with $p\not=q$ in general. A counterexample is provided by the quadratic Hamiltonian $H(x,\xi)=|x|^2$, $(x,\xi)\in\rdd$. In this case $U_N(t)=U(t)$ and  the exact solution $U(t)$ is the multiplication operator $U(t)f(x)=e^{i|x|^2} f(x)$. It was proved in \cite[Proposition 7.1]{fio1} that this operator is not continuous on $M^{p,q}(\Ren)$, with $p\not=q$.\par
(iii) For  $s\in\bR$, $p=q=2$, recall that $M^{2}_{v_s}(\Ren)=\mathcal{Q}_s$, the Shubin-Sobolev spaces \cite{BCG, Shubin91}. Hence the results of the previous theorem hold also for semi-classical Shubin-Sobolev spaces.
\end{remark}

\section*{Acknowledgments} This research was partially supported by  the Gruppo
Nazionale per l'Analisi Matematica, la Probabilit\`a e le loro
Applicazioni (GNAMPA) of the Istituto Nazionale di Alta Matematica
(INdAM), Italy.

\end{document}